\documentclass[conference]{IEEEtran}

\usepackage{graphicx} 
\usepackage[caption=false]{subfig}

\usepackage{amsmath,amssymb,amsfonts} 
\usepackage{bm, bbm} 
\usepackage{amsthm} 
\usepackage{algorithm} 
\usepackage[noend]{algpseudocode} 

\DeclareMathOperator*{\argmax}{arg\,max}

\usepackage{tcolorbox} 
\usepackage{setspace}
\usepackage{adjustbox}

\usepackage{hyperref}
\hypersetup{
    colorlinks = true,
    linkcolor = black,
    urlcolor = black,
    citecolor = black
}

\usepackage{cite} 
\usepackage[capitalise]{cleveref}
\crefname{appsec}{Appendix}{Appendices}
\crefname{equation}{}{}
\usepackage{jabbrv}

\newtheorem{theorem}{Theorem}

\newtheorem{lemma}[theorem]{Lemma}
\newtheorem{prop}{Proposition}

\usepackage{xcolor}
\usepackage{color} 
\definecolor{darkred}{rgb}{0.6,0.0,0.0}
\definecolor{darkgreen}{rgb}{0,0.50,0}
\definecolor{lightblue}{rgb}{0.0,0.42,0.91}
\definecolor{orange}{rgb}{0.99,0.48,0.13}
\definecolor{grass}{rgb}{0.18,0.80,0.18}
\definecolor{pink}{rgb}{0.97,0.15,0.45}
\definecolor{mypagecolor}{RGB}{249,249,249}
\pagecolor{mypagecolor}

\usepackage[acronym, toc]{glossaries}
\newacronym{MEDA}{MEDA}{McGill Engineering Doctoral Award}
\newacronym{RB}{RB}{resource block} 
\newacronym{IoT}{IoT}{Internet of things}
\newacronym{MTC}{MTC}{machine-type communication}
\newacronym{mMTC}{mMTC}{massive machine-type communication}
\newacronym{IIoT}{IIoT}{industrial Internet of Things}
\newacronym{HTC}{HTC}{Human-type communication}
\newacronym{MA}{MA}{multiple access}
\newacronym{PA}{PA}{priority access}
\newacronym{RACH}{RACH}{random access channel}
\newacronym{PRACH}{PRACH}{pysical random access channel}

\newacronym{PG}{PG}{power grid}
\newacronym[\glslongpluralkey={smart grids}, \glsshortpluralkey={SGs}]{SG}{SG}{smart grid} 
\newacronym{MG}{MG}{micro grids}
\newacronym{DOE}{DOE}{Department of Energy}
\newacronym{AMI}{AMI}{advanced metering infrastructure}
\newacronym{AMR}{AMR}{automatic meter reading}
\newacronym{DLC}{DLC}{direct load control}
\newacronym{SM}{SM}{smart meter}
\newacronym{HAN}{HAN}{home area network}
\newacronym{DR}{DR}{demand response}
\newacronym{EV}{EV}{electric vehicle}
\newacronym{WASA}{WASA}{wide-area situational awareness}
\newacronym{PMU}{PMU}{phasor measurement unit}
\newacronym{DER}{DER}{distributed energy resources}
\newacronym{DGM}{DGM}{distributed grid management}
\newacronym{DA}{DA}{distribution automation}
\newacronym{SA}{SA}{substation automation}
\newacronym{NIST}{NIST}{National Institute of Standards and Technology}
\newacronym{FERC}{FERC}{Federal Energy Regulatory Commission}
\newacronym{PV}{PV}{photovoltaic}
\newacronym{SCADA}{SCADA}{supervisory control and data acquisition}
\newacronym{U.S.}{U.S.}{United States}

\newacronym{UE}{UE}{user equipment}
\newacronym{CB-MA}{CBMA}{contention-based multiple access}
\newacronym{CF-MA}{CFMA}{conflict-free multiple access}
\newacronym{FDMA}{FDMA}{frequency division multiple access}
\newacronym{TDMA}{TDMA}{time division multiple access}
\newacronym{CDMA}{CDMA}{code division multiple access}
\newacronym{SDMA}{OMA}{spatial division multiple access}
\newacronym{LDS-CDMA}{LDS-CDMA}{low density signature code division multiple access}
\newacronym{PDMA}{PDMA}{pattern division multiple access}
\newacronym{MUSA}{MUSA}{multi-user shared access}
\newacronym{SCMA}{SCMA}{sparce code multiple access}
\newacronym{IDMA}{IDMA}{interleave division multiple access}
\newacronym{SPMA}{SPMA}{statistical based multiple access}
\newacronym{DS-CDMA}{DS-CDMA}{direct sequence code division multiple access}
\newacronym{FH-CDMA}{FH-CDMA}{frequency hopping code division multiple access}
\newacronym{IRSA}{IRSA}{irregular repetition slotted ALOHA}

\newacronym{OMA}{OMA}{orthogonal multiple access}
\newacronym{NOMA}{NOMA}{non-orthogonal multiple access}
\newacronym{CD-NOMA}{CD-NOMA}{code domain non-orthogonal multiple access}
\newacronym{PD-NOMA}{PD-NOMA}{power domain non-orthogonal multiple access}
\newacronym{MIMO}{MIMO}{multiple input multiple output}
\newacronym{QAM}{QAM}{quadrature amplitude modulation}
\newacronym{FEC}{FEC}{forward error correction}
\newacronym{MAC}{MAC}{medium access control}
\newacronym{SIC}{SIC}{successive interference cancellation}
\newacronym{BS}{BS}{base station}
\newacronym{LTE}{LTE}{long-term evolution}
\newacronym{LTE-A}{LTE-A}{long-term evolution advanced}

\newacronym{SE}{SE}{spectral efficiency}
\newacronym{CSI}{CSI}{channel state information}
\newacronym{SUD}{SUD}{single-user detection}
\newacronym{MUD}{MUD}{multi-user detection}
\newacronym{MAP}{MAP}{maximum \textit{a posteriori}}
\newacronym{ML}{ML}{maximum likelihood}
\newacronym{RCML}{RCML}{reduced complexity maximum likelihood}
\newacronym{i.i.d.}{i.i.d.}{independent and identically distributed}

\newacronym{FDD}{FDD}{frequency division duplexing}
\newacronym{TDD}{TDD}{time division duplexing}
\newacronym{PSD}{PSD}{power spectral density}
\newacronym{MAI}{MAI}{multiple access interference}
\newacronym{QoS}{QoS}{quality of service}
\newacronym{TTNT}{TTNT}{Tactical Targeting Network Technology}
\newacronym{COS}{COS}{channel occupancy statistic}
\newacronym{PPP}{PPP}{Poisson point process}
\newacronym{3GPP}{3GPP}{3rd Generation Partnership Project}

\newacronym{H-UE}{H-UE}{high-priority user equipment}
\newacronym{L-UE}{L-UE}{low-priority user equipment}
\newacronym{ACB}{ACB}{access class barring}
\newacronym{D-ACB}{D-ACB}{dynamic access class barring}

\newacronym{QPAC}{QPAC}{queue-aware priority access classification}
\newacronym{PRADA}{PRADA}{prioritized random access with dynamic access barring}
\newacronym{PLPS}{PLPS}{priority-based load adaptive preamble separation}
\newacronym{CRA}{CRA}{Coded Random Access}

\newacronym{RHS}{RHS}{right hand side}
\newacronym{LHS}{LHS}{left hand side}
\newacronym{MAE}{MAE}{mean absolute error}
\newacronym{MSE}{MSE}{mean squared error}
\newacronym{SQP}{SQP}{sequential quadratic programming}

\newacronym{RL}{RL}{reinforcement learning}
\newacronym{DRL}{DRL}{distributed reinforcement learning}
\newacronym{MDP}{MDP}{Markov decision process}
\newacronym{MAB}{MAB}{multi-armed bandit}
\newacronym{UCB}{UCB}{upper confidence bound}
\newacronym{AS}{AS}{action space}
\newacronym{e-greedy}{$\epsilon$-greedy}{epsilon-greedy}
\newacronym{DEG}{DEG}{decaying epsilon-greedy}
\newacronym{DQN}{DQN}{deep Q-network}
\newacronym[\glsshortpluralkey={algs.}]{alg}{alg.}{algorithm}
\newacronym{CE}{CE}{cross entropy}
\newacronym{PPO}{PPO}{proximal policy optimization}
\newacronym{SARSA}{SARSA}{state-action-reward-state-action}
\newacronym{CMMPP}{CMMPP}{Coupled Markov Modulated Poisson Processes}

\newacronym{PSS}{PSS}{primary synchronization signal}
\newacronym{SSS}{SSS}{secondary synchronization signal}
\newacronym{MIB}{MIB}{master information block}
\newacronym{SIB}{SIB}{system information block}
\newacronym{SIB2}{SIB2}{system information block 2}
\newacronym{PDCCH}{PDCCH}{physical downlink control channel}
\newacronym{ZC}{ZC}{Zadoff-Chu}
\newacronym{TA}{TA}{timing advance}
\newacronym{CP}{CP}{cyclic prefix}
\newacronym{GT}{GT}{guard time}
\newacronym{RA}{RA}{random access}
\newacronym{5G}{5G}{fifth generation}
\newacronym{NR}{NR}{new radio}
\newacronym{RAR}{RAR}{random access response}
\newacronym{RRC}{RRC}{radio resource control}
\newacronym{DCI}{DCI}{downlink control information}
\newacronym{PUSCH}{PUSCH}{physical uplink shared channel}
\newacronym{PUCCH}{PUCCH}{physical uplink control channel}
\newacronym{PHICH}{PHICH}{physical hybridARQ indicator channel}
\newacronym{PCFICH}{PCFICH}{physical control format indicator channel}
\newacronym{OFDM}{OFDM}{orthogonal frequency division multiplexing}
\newacronym{OFDMA}{OFDMA}{orthogonal frequency division multiple access}
\newacronym{RE}{RE}{resource element}
\newacronym{CBRA}{CBRA}{contention-based random access}
\newacronym{CFRA}{CFRA}{contention-free random access}
\newacronym{PAPR}{PAPR}{peak-to-average power ratio}
\newacronym{GCL}{GCL}{generalized chirp-like}
\newacronym{PN}{PN}{pseudorandom noise}
\newacronym{TTI}{TTI}{transmission time interval}
\newacronym{RA-RNTI}{RA-RNTI}{Random Access Radio Network Temporary Identifier}
\newacronym{C-RNTI}{C-RNTI}{Cell Radio Network Temporary Identifier}

\usepackage{textcomp}

\def\BibTeX{{\rm B\kern-.05em{\sc i\kern-.025em b}\kern-.08em T\kern-.1667em\lower.7ex\hbox{E}\kern-.125emX}}
\begin{document}
\makeatletter
\def\endthebibliography{%
    \def\@noitemerr{\@latex@warning{Empty `thebibliography' environment}}%
    \endlist
}
\makeatother

\title{Load Estimation in a Two-Priority mMTC Random Access Channel}

\author{\IEEEauthorblockN{Ahmed O. Elmeligy}
\IEEEauthorblockA{\textit{McGill University} \\
ahmed.elmeligy@mail.mcgill.ca}
\and
\IEEEauthorblockN{Ioannis Psaromiligkos}
\IEEEauthorblockA{\textit{McGill University} \\
ioannis.psaromiligkos@mcgill.ca}
\and
\IEEEauthorblockN{Au Minh}
\IEEEauthorblockA{\textit{Hydro-Quebéc Research Institute (IREQ)} \\
au.minh2@hydroquebec.com}
}

\maketitle

\begin{abstract}
    The use of cellular networks for massive machine-type communications (mMTC) is an appealing solution due to the wide availability of cellular infrastructure.
    
    Estimating the number of devices (network load) is vital for efficient allocation of the available resources, especially for managing the random access channel (RACH) of the network.
    This paper considers a two-priority RACH and proposes two network load estimators: a maximum likelihood (ML) estimator and a reduced complexity (RCML) variant. 
    The estimators are based on a novel model of the random access behavior of the devices coupled with a flexible analytical framework to calculate the involved probabilities.
    Monte Carlo simulations demonstrate the accuracy of the proposed estimators for different network configurations.
    Results depict increased estimation accuracy using non-uniform preamble selection probabilities compared to the common uniform probabilities at no extra computational cost.
\end{abstract}

\section{Introduction}
\label{sec: introduction} 
The use of already deployed cellular infrastructure is an appealing solution to providing wireless connectivity for \gls{mMTC} applications~\cite{leyva2016performance}.
At the same time, combining multi-priority applications in a single network is a promising cost-efficient approach as it avoids deploying multiple networks.
However, servicing the massive number of devices in \gls{mMTC} poses a significant challenge to the cellular network, particularly, by congesting and overloading the \gls{RACH}~\cite{condoluci2015toward, wiriaatmadja2014hybrid}.
The \gls{RACH}, which is usually available periodically, is the channel the devices use to request access to the network~\cite{wei2014modeling}.
Each time the \gls{RACH} is available is known as a \gls{RACH} slot, and at each slot, a device (henceforth referred to as a \gls{UE}) randomly chooses a preamble from a finite set of preambles and transmits it to the \gls{BS}; if a preamble is selected by only one device, the \gls{BS} can successfully decode it.
On the other hand, if two or more devices choose the same preamble a collision happens~\cite{arouk2016accurate}.

The network load is defined as the number of \glspl{UE} the network serves, and the ratio of \glspl{UE} to the number of preambles is known as the overloading factor; a high overloading factor indicates that the network is congested, thus increasing the number of collisions and reducing the network throughput.
Knowing the load is vital for the network operator to decide how to allocate the available finite resources to the \gls{UE} pool.

Several methods have been proposed to estimate the network load in the \gls{RACH}.
A closed-form expression for the joint probability of the number of successful and collided \gls{UE} preamble transmissions within a \gls{RACH} slot is derived in~\cite{tello2018efficient}.
\Gls{ML} estimation and Bayesian techniques are used to estimate the network load.
Two algorithms are proposed based on whether the number of successes and collisions are known at the \gls{BS} or only the number of successes is known.
However, the probability of unselected preambles is not considered, which is additional information known to the \gls{BS} and can be used to improve the estimation accuracy.

A combinatorial model presented in~\cite{wei2014modeling} investigates the transient behavior of \glspl{RACH} with bursty arrivals.
The model obtains the average number of successful and collided \glspl{UE} in a single \gls{RACH} slot.
The authors then estimate the number of contending \glspl{UE} in each slot.
Simulations show that the estimation is accurate in the case of a high \gls{UE} count but suffers for few \gls{UE}. 

The number of active \glspl{UE} in \gls{IRSA} access protocol is estimated in~\cite{sun2018detecting}; \gls{SIC} is adopted at the receiver to recover the transmitted packets.
The authors propose a \gls{MAP} detector for the number of unrecovered \glspl{UE} in a specific \gls{SIC} iteration.
A sub-optimal detector is also proposed to reduce the computational complexity via approximations.
Numerical results illustrate that the suboptimal detector's \gls{MAE} increases with the overloading.

In~\cite{duan2016d}, an estimate of the network load is used in two dynamic \gls{ACB} algorithms that operate without the \gls{BS} knowing the number of devices.
The first algorithm optimizes the \gls{ACB} factor for a fixed number of preambles, where the \gls{ACB} factor determines the backoff time of collided devices.
The second algorithm optimizes the \gls{ACB} factor and the number of available preambles.

Finally, a backoff scheme is developed by~\cite{access2021medium} that uses a backoff indicator to determine a random waiting period for the collided \glspl{UE}.
The backoff indicator is dynamically adjusted in~\cite{althumali2020dynamic} based on the availability of resources and the number of backlogged \glspl{UE}.

In all these works, only a single \gls{UE} priority class is considered; hence superimposing numerous applications with varying \gls{QoS} requirements in the same network to save resources would not be possible.
Additionally, only a single \gls{RACH} slot is used, and the estimation techniques are not extended to multiple \gls{RACH} slots, which could significantly improve the estimator's accuracy. 

Several additional works consider single and multi-UE priority classes, focusing on improving the \gls{RACH} performance by first introducing a performance metric, and then optimizing the \gls{RACH} parameters, such as the number of preambles, the \gls{RACH} periodicity, or the \gls{ACB} factor, to maximize the metric~\cite{chowdhury2022queue, althumali2022priority, tello2017performance, alvi2021performance, duan2013dynamic}. 
However, no network load estimation techniques are proposed, and the number of contending \glspl{UE} is assumed to be known.

In this paper, we propose a load estimation method that addresses the above shortcomings. 
Specifically, the main contributions of this paper are as follows:
\begin{itemize}
    \item We present a novel approach to model a two-priority \gls{RACH}, which allows us to define access patterns that describe the random access behavior of \glspl{UE} as observed by the \gls{BS}.
    \item We develop an analytical framework to calculate the probability of observing a given access pattern.
    \item Our approach allows for non-uniform preamble selection probabilities. To our knowledge, the literature lacks a model that considers non-uniform preamble selection, which can provide greater flexibility when allocating resources to different \gls{UE} priority classes.
    \item We propose an \gls{ML} estimator of the network load from observed access patterns over multiple \gls{RACH} slots.
    \item We formulate a \gls{RCML} estimator that discards a portion of the available information at the \gls{BS}.
    \item Our results show an increase in estimation accuracy using non-uniform preamble selection probabilities at no additional computational cost relative to the traditional uniform approach.
\end{itemize}
Monte Carlo simulation results validate the analytical framework and showcase the proposed estimators' accuracy for various system setups.

The rest of this paper is organized as follows:
\cref{sec: system_model} presents the system model that formulated the observations of the \gls{RACH} at the \gls{BS} as patterns. 
The \gls{ML} and  \gls{RCML} estimators are developed in \cref{sec: proposed_method} along with the algorithm that computes the pattern probabilities. 
\cref{sec: simulation} presents the simulation setup and numerical results that showcase the estimators' accuracies. 
Finally, \cref{sec: conclusion} concludes the paper.

\section{System Model and Problem Statement}
\label{sec: system_model}
We consider a system model with $n^h$ \glspl{H-UE} and $n^l$ \glspl{L-UE}.
The \glspl{H-UE} and \glspl{L-UE} differ in \gls{QoS} requirements, whereas the \glspl{H-UE} have stricter requirements, such as lower latency and higher data rates.
Note that $n^h$ and $n^l$ are fixed and unknown to the \gls{BS}. 
Transmitting preambles over the \gls{RACH} can be viewed as having a finite number of \glspl{RB} randomly accessed by the pool of \glspl{UE}.
Time is discretized into slots, indexed by $t$, and each slot is divided into $M$ \glspl{RB}.
We assume that at slot $t$, each of the $n = n^h + n^l$ \glspl{UE} randomly chooses an \gls{RB}.
The probability that a \gls{H-UE} chooses \gls{RB} $i$ is $p^h_i$; likewise, the probability that an \gls{L-UE} chooses \gls{RB} $i$ is $p^l_i$. 
Let $\bm{p}^h = [p^h_1, \dots, p^h_M]$ be the vector collecting the \gls{RB} selection probabilities for an \gls{H-UE}, while $\bm{p}^l = [p^l_1, \dots, p^l_M]$ be the corresponding vector for an \gls{L-UE}.
Clearly, we have $\sum_{i=1}^M p^h_i = 1$ and $\sum_{i=1}^M p^l_i = 1$. 

For each of the $M$ \glspl{RB}, one of the four events listed below may occur:
\begin{enumerate}
    \item Event $h$: A single \gls{H-UE} selects the \gls{RB}.
    \item Event $l$: A single \gls{L-UE} selects the \gls{RB}.
    \item Event $\phi$: No \gls{UE} selects the \gls{RB}, i.e., the \gls{RB} is unoccupied or empty.
    \item Event $x$: More than one \gls{UE} selects the \gls{RB}, i.e., a collision occurs.
\end{enumerate}
Each of these events can be detected by the \gls{BS}~\cite{zhang2022ppo}.
However, in the case of a collision, the \gls{BS} does not know how many \glspl{UE} are involved~\cite{wu2012fast, galinina2013stabilizing, wang2015optimal} nor their priority. 
Thus, at each time slot $t$, the \gls{BS} observes an access pattern, $\bm{\pi}_t \in \{h,l,\phi,x\}^M$, that is the sequence of events that occur across the $M$ \glspl{RB}; \cref{fig: patterns} shows an example of such patterns observed over $T$ time slots.

\setlength{\textfloatsep}{2.5pt}
\setlength{\floatsep}{2.5pt}
\setlength{\intextsep}{2.5pt}
\begin{figure}
    \centering
    \begin{adjustbox}{max width=0.35\textwidth}
        \includegraphics{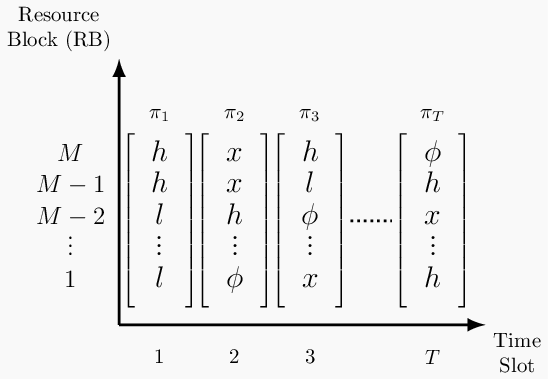}
    \end{adjustbox}
    \caption{Example of a pattern sequence over $T$ time slots.}
    \label{fig: patterns}
\end{figure}

The problem we consider in this paper is to estimate the network load defined as the number of \glspl{H-UE}, $n^h$, and \glspl{L-UE}, $n^l$, by observing a sequence of $T$ access patterns, $\bm{\pi}_1, \bm{\pi}_2, \ldots, \bm{\pi}_T$.
We assume that the preamble selection probability vectors $\bm{p}^h$ and $\bm{p}^l$ are known at the \gls{BS}, and that $\bm{\pi}_1, \bm{\pi}_2, \ldots, \bm{\pi}_T$ are independent and identically distributed (i.i.d.).
This is a common assumption~\cite{sun2018detecting}, meant to simplify the system model. Extension of the model to the non-i.i.d. case is possible through the use of conditional probabilities.
Such extension, however, relies heavily on the implementation of the RACH protocol, e.g., number of retries, back-off time, and is beyond the scope of this paper.

\section{Proposed Method}
\label{sec: proposed_method}
This section is divided into thee parts. First, we propose a \gls{ML} estimation technique to estimate the network load given the observed patterns $\bm{\pi}_1, \bm{\pi}_2, \ldots, \bm{\pi}_T$.
As the proposed algorithm requires the evaluation of the probabilities of the observed patterns,  
in \cref{subsec: pattern probability} we present an algorithm to compute these probabilities as a function of the number of \glspl{H-UE} and \glspl{L-UE}.
Finally, we develop an \gls{RCML} estimator that discards a portion of the available information at the \gls{BS} to reduce the computational complexity.

\subsection{Maximum Likelihood Estimation}
\label{subsec: mle}
Let us denote by $P(\,\cdot\,; n^h, n^l)$ the probability of observing a given pattern or sequence of patterns, which, of course, depends on $n^h$ and $n^l$. 
Then, the \gls{ML} estimator is given by:
\begin{align}
    \label{eq: load estimator}
    \hat{n}^h, \hat{n}^l &= 
    \argmax_{n^h, n^l} P(\bm{\pi}_1, \bm{\pi}_2, \ldots, \bm{\pi}_T; n^h, n^l) \nonumber \\
    &=\argmax_{n^h, n^l} \prod_{t=1 }^TP(\bm{\pi}_t; n^h, n^l)
\end{align}
where the last equality holds due to the i.i.d. assumption on $\bm{\pi}_1, \bm{\pi}_2, \ldots, \bm{\pi}_T$. 

\subsection{Derivation of Pattern Probabilities}
\label{subsec: pattern probability}
We start with an alternative representation of an access pattern, $\bm{\pi}$.
Note that the subscript $t$ and the dependence on $n^h$ and $n^l$ are omitted throughout this section to avoid overloading the notation.
Define $H_i$ as the event that one \gls{H-UE} occupies \gls{RB} $i$.
Similarly, define $L_i$, $\Phi_i$, and $X_i$ as the events that one \gls{L-UE} occupies \gls{RB} $i$, no \gls{UE} occupies \gls{RB} $i$ (RB $i$ is empty), and two or more \glspl{UE} occupy \gls{RB} $i$ (collision occurs), respectively.
Let set $\mathcal{H}$ contain the indices of the \glspl{RB} that contain only one \gls{H-UE}.
Likewise, sets $\mathcal{L}$, $\mathbf{\mathit{\Phi}}$, and $\mathcal{X}$ contain the indices of the \glspl{RB} that contain only one \gls{L-UE}, are empty, and have collisions, respectively.
Hence, a pattern can be written as the intersection of the events $H_i$ for $i \in \mathcal{H}$, $L_i$ for $i \in \mathcal{L}$, $\Phi_i$ for $i \in \mathbf{\mathit{\Phi}}$, and $X_i$ for $i \in \mathcal{X}$:

\begin{equation}
    \label{eq: pattern sequence}
    \bm{\pi} = \overline{H}, \overline{L}, \overline{\Phi}, \overline{X}
\end{equation}

where $\overline{H}=\bigcap_{i\in{\mathcal H}}H_i$, $\overline{L}=\bigcap_{i\in{\mathcal L}}L_i$, $\overline{\Phi}=\bigcap_{i\in\mathbf{\mathit{\Phi}}}\Phi_i$, and $\overline{X}=\bigcap_{i\in{\mathcal X}}X_i$
In \cref{fig: pattern example}, we show an example for a pattern of length $M=8$ along with
the corresponding sequence of events $H_i$, $L_i$, $\Phi_i$, and $X_i$, as well as the index sets $\mathcal{H}$, $\mathcal{L}$, $\mathbf{\mathit{\Phi}}$, and $\mathcal{X}$.
\setlength{\textfloatsep}{2.5pt}
\setlength{\floatsep}{2.5pt}
\setlength{\intextsep}{2.5pt}
\begin{figure}
    \centering
    \begin{adjustbox}{max width=0.38\textwidth}
        \includegraphics{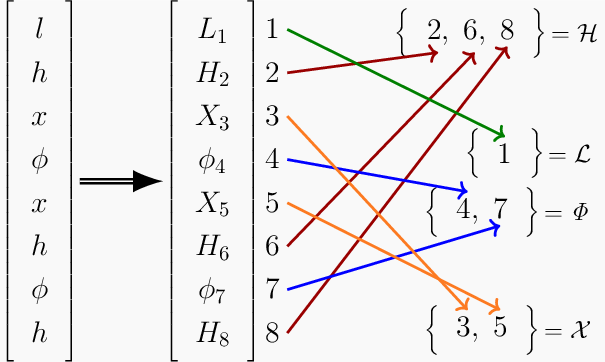}
    \end{adjustbox}
    \caption{Example pattern of length $M=8$. In slots 2, 6, and 8 we have collision-free \gls{H-UE} transmissions, denoted as events $H_2$, $H_6$ and $H_8$; hence, ${\mathcal H} = \{2, 6, 8\}$. Similarly, the event in slot 1 a collision-free \gls{L-UE} transmission, the events in slots 4 and 7 no transmissions, and the events in slots 3 and 5 signify collisions.}
    \label{fig: pattern example}
\end{figure}

From \cref{eq: pattern sequence}, we can decompose $P(\bm{\pi})$ as follows:

\begin{equation}
    \label{eq: conditional pattern probability}
    P(\bm{\pi}) = 
    P\left(\overline{H}\right)
    P\left(\overline{L} \middle| \overline{H}\right)
    P\left(\overline{\Phi} \middle| \overline{H}, \overline{L}\right)
    P\left(\overline{X} \middle| \overline{H}, \overline{L}, \overline{\Phi}\right)
\end{equation}

Next, we derive the probabilities in the \gls{RHS} of \cref{eq: conditional pattern probability}.

The first factor in \cref{eq: conditional pattern probability}, $P\left(\overline{H}\right)$, is the probability that each \gls{RB} in $\mathcal{H}$ is occupied by one \gls{H-UE}. Let $\mathcal{H} = \{h_1, h_2, \dots, h_H\}$. Then we have:
\begin{multline}
    \label{eq: P(H)}
    P\left(\overline{H}\right) =
    P(H_{h_1})P(H_{h_2}|H_{h_1}) \times \dots \times \\
    P(H_{h_H}|H_{h_1}, \dots, H_{h_{H-1}})
\end{multline}

The probabilities in \cref{eq: P(H)} can be calculated using the following proposition. 
    \begin{prop}
        \label{prop: P(H)}
        The first factor in the \gls{RHS} of \cref{eq: P(H)} is given by:
        \begin{equation}
            \label{eq: P(h_1)}
            P(H_{h_1}) = n^h p^h_{h_1} (1-p^h_{h_1})^{n^h-1} (1-p^l_{h_1})^{n^l}
        \end{equation}
        The rest of the factors are given by:
        \begin{multline}
            \label{eq: P(h_k|h_1 ... h_{k-1})}
            P(H_{h_k}|H_{h_1}, \dots, H_{h_{k-1}}) = \\
            N^h_k \hat{p}^h_{h_k} (1-\hat{p}^h_{h_k})^{N^h_k-1} (1-\hat{p}^l_{h_k})^{n^l}
        \end{multline}
        where $N^h_k$, $\hat{p}^h_{h_k}$, and $\hat{p}^l_{h_k}$ are calculated as follows:
        \begin{align}
        \label{eq: normalization P(H)0}
        N^h_k &= n^h -(k-1)\\
        \label{eq: normalization p(H)1}
        \hat{p}^h_{h_k} &= \frac{p^h_{h_k}}{p^h_{h_k}+\ldots+p^h_{H} + \sum_{i \in \mathcal{L}, \mathbf{\mathit{\Phi}}, \mathcal{X}} p^h_{i}} \\
        \label{eq: normalization p(H)2}
        \hat{p}^l_{h_k} &= \frac{p^l_{h_k}}{p^l_{h_k}+\ldots+p^l_{H} + \sum_{i \in \mathcal{L}, \mathbf{\mathit{\Phi}}, \mathcal{X}} p^l_{i}}
        \end{align}
    \end{prop}
The proof of \cref{prop: P(H)} is given in Appendix~\ref{app: P(H) proof}.
$P\left(\overline{H}\right)$ can be calculated recursively using 
\cref{alg: P(H)} also given in Appendix~\ref{app: rec algs}.

Following the same procedure, the second factor in \cref{eq: conditional pattern probability}, $P\left(\overline{L} \middle| \overline{H}\right)$, is the probability that each \gls{RB} in $\mathcal{L}$ is occupied by one \gls{L-UE} given that 
the \glspl{RB} in $\mathcal{H}$ are unavailable. 
Let $\mathcal{L}=\{l_1, l_2, \dots, l_L\}$.
Then we have:
\begin{multline}
    \label{eq: P(L|H)}
    P\left(\overline{L} \middle| \overline{H}\right) =
    P\left(L_{l_1} \middle| \overline{H}\right) 
    P\left(L_{l_2} \middle| L_{l_1}, \overline{H}\right) \times \dots \times  \\
    P\left(L_{l_L} \middle| L_{l_1}, \dots, L_{l_{L-1}}, \overline{H}\right)
\end{multline}

The probabilities in \cref{eq: P(L|H)} are calculated using \cref{prop: P(L|H)} .
    \begin{prop}
        \label{prop: P(L|H)}
        The first factor of the \gls{RHS} of \cref{eq: P(L|H)} is given by:
        \begin{equation}
            \label{eq: P(l_1|H)}
            P\left(L_{l_1} \middle| \overline{H}\right) = n^l \hat{p}^l_{l_1} (1-\hat{p}^l_{l_1})^{n^l-1} (1-\hat{p}^h_{l_1})^{N^h}
        \end{equation}
        
        The rest of the factors are given by:
        \begin{equation}
            \label{eq: P(l_k|l_1 ... l_{k-1} H)}
            P\left(L_{l_k} \middle| L_{l_1}, \dots, L_{l_{k-1}}, \overline{H}\right) =
            N^l_k \hat{p}^l_{l_k} (1-\hat{p}^l_{l_k})^{N^l_k-1} (1-\hat{p}^h_{l_k})^{N^h}
        \end{equation}
        
        where $N^h$, $N^l_k$, $\hat{p}^h_{l_k}$, and $\hat{p}^l_{l_k}$ are as follows:
        \begin{align}
        N^h &= n^h - H, \,\,  N^l_k = n^l - (k-1) \\ 
        \label{eq: normalization P(L|H)2}
        \hat{p}^h_{l_k} &= \frac{p^h_{l_k}}{p^h_{l_k}+\ldots+p^h_{L} + \sum_{i \in \mathbf{\mathit{\Phi}}, \mathcal{X}} p^h_{i}} \\
        \label{eq: normalization P(L|H)3}
        \hat{p}^l_{l_k} &= \frac{p^l_{l_k}}{p^l_{l_1}+\ldots+p^l_{L} + \sum_{i \in \mathbf{\mathit{\Phi}}, \mathcal{X}} p^l_{i}}
        \end{align}
    \end{prop}
The proof of \cref{prop: P(L|H)} follows the same procedure as \cref{prop: P(H)} and is omitted due to lack of space.  \cref{alg: P(L|H)} in Appendix~\ref{app: rec algs} shows how $P\left(\overline{L} \middle| \overline{H}\right)$ is calculated.

Moving on to the third factor in \cref{eq: conditional pattern probability}, $P\left(\overline{\Phi} \middle| \overline{H}, \overline{L}\right)$ is the probability that no \glspl{UE} occupy $\mathbf{\mathit{\Phi}}$ given that 
the \glspl{RB} in $\mathcal{H}$ and $\mathcal{L}$ are unavailable.
Let $\mathbf{\mathit{\Phi}} = \{\phi_1, \phi_2, \dots, \phi_\Phi\}$.
Then we have:
\begin{multline}
    \label{eq: P(E|H L)}
    P\left(\overline{\Phi} \middle| \overline{H}, \overline{L}\right) = 
    P\left(\Phi_{\phi_1} \middle| \overline{H}, \overline{L}\right)
    P\left(\Phi_{\phi_2} \middle| \Phi_{\phi_1}, \overline{H}, \overline{L} \right) \times\dots \times \\
    P\left(\Phi_{\phi_\Phi} \middle|\Phi_{\phi_1}, \dots, \Phi_{\phi_{\Phi-1}}, \overline{H}, \overline{L} \right)
\end{multline}

The probabilities in \cref{eq: P(E|H L)}  are calculated using \cref{prop: P(E|H L)} the proof of which is omitted due to lack of space.
    \begin{prop}
        \label{prop: P(E|H L)}
        The first factor of the \gls{RHS} of \cref{eq: P(E|H L)} is given by:
        \begin{equation}
            \label{eq: P(phi_1|H L)}
            P\left(\Phi_{\phi_1} \middle| \overline{H}, \overline{L}\right) =
            (1-\hat{p}^h_{\phi_1})^{N^h} (1-\hat{p}^l_{\phi_1})^{N^l}
        \end{equation}
        The rest of the factors are given by:
        \begin{equation}
            \label{eq: P(phi_k|phi_1 ... phi_{k-1} H L)}
            P\left(\Phi_{\phi_k} \middle| \Phi_{\phi_1}, \dots, \Phi_{\phi_{k-1}}, \overline{H}, \overline{L}\right) = 
            (1-\hat{p}^h_{\phi_k})^{N^h} (1-\hat{p}^l_{\phi_k})^{N^l}
        \end{equation}
        where $N^h$, $N^l$, $\hat{p}^h_{l_k}$, and $\hat{p}^l_{l_k}$ are as follows:
        \begin{align}
            N^h &= n^h - H, 
            \,\, N^l_k = n^l - L \\
            \label{eq: normalization P(E|H L)2}
            \hat{p}^h_{\phi_k} &= \frac{p^h_{\phi_k}}{p^h_{\phi_k}+\ldots+p^h_{\Phi} + \sum_{i \in \mathbf{\mathcal{X}}} p^h_{i}} \\
            \label{eq: normalization P(E|H L)3}
            \hat{p}^l_{\phi_k} &= \frac{p^l_{\phi_k}}{p^l_{\phi_k}+\ldots+p^l_{\Phi} + \sum_{i \in \mathbf{\mathcal{X}}} p^l_{i}}
        \end{align}
    \end{prop}
The algorithm that calculates the third factor in \cref{eq: conditional pattern probability} is similar to \cref{alg: P(L|H)}, hence it has been omitted for brevity.

Finally, the last factor in \cref{eq: conditional pattern probability}, $P\left(\overline{X} \middle| \overline{H}; \overline{L} \; \overline{\Phi}\right)$, is the probability that two or more \glspl{UE} occupy each \gls{RB} in $\mathcal{X}$.
Let $\mathcal{X}=\{x_1, x_2, \dots, x_X\}$.
\cref{prop: P(X|H L E)} depicts how $P\left(\overline{X} \middle| \overline{H}, \overline{L}, \overline{\Phi}\right)$ is calculated.
    \begin{prop}
        \label{prop: P(X|H L E)}
        $P\left(\overline{X} \middle| \overline{H}, \overline{L}, \overline{\Phi}\right)$ is given by:
        \begin{subequations}
            \begin{multline}
                \label{eq: P(X RBs|H L E)}
                P\left(\overline{X} \middle| \overline{H}, \overline{L}, \overline{\Phi}\right) = \\
                \sum_{k_1=2}^{N^{h}_{1} + N^{l}_{1}} \cdots \sum_{k_X=2}^{N^{h}_{X} + N^{l}_{X}} \sum_{i_1=0}^{N^{h}_{1}} \cdots \sum_{i_X=0}^{N^{h}_{X}} \\
                \prod_{j=1}^{X} \binom{N^h}{i_j} (\hat{p}^h_{x_j})^{i_j} (1-\hat{p}^h_{x_j})^{N^h-i_j} \times \\\binom{N^l}{k_j-i_j} (\hat{p}^l_{x_j})^{k_j-i_j} (1-\hat{p}^l_{x_j})^{N^l-k_j+i_j}
            \end{multline}
            \begin{flalign}
                \hfill && \textrm{s.t.} \quad &k_m-i_m \geq 0, \quad m \in \mathcal{X}\\
                \hfill && &\sum_{m=1}^X{i_m} = N^h, \qquad \sum_{m=1}^X{k_m} = N^l - N^h
            \end{flalign}
        \end{subequations}
        where $N^{h}_{j}$, $N^{l}_{j}$, $\hat{p}^{h}_{x_j}$ and $\hat{p}^{l}_{x_j}$ are as follows:
        \begin{align}
            \label{eq: normalization P(X|H L E)0}
            N^{h}_{j} &= n^h - H - i_{j-1} - (j-1) \text{, with } i_0=0 \\
            \label{eq: normalization P(X|H L E)1}
            N^{l}_{j} &= n^l - L - (k_{j-1}-i_{j-1}) - (j-1) \text{, with } k_0=i_0=0 \\
            \hat{p}^h_{x_j} &= \frac{p^h_{x_j}}{p^h_{x_j}+\ldots+p^h_{X}}, 
            \,\, \hat{p}^l_{x_j} = \frac{p^l_{x_j}}{p^l_{x_j}+\ldots+p^l_{X}}
        \end{align}      
    \end{prop}
The proof of \cref{prop: P(X|H L E)} is omitted due to lack of space.
\cref{alg: P(X|H L E) recursive} in Appendix~\ref{app: rec algs} shows how $P\left(\overline{X} \middle| \overline{H}, \overline{L}, \overline{\Phi}\right)$ is calculated.

\subsection{Reduced Complexity Maximum Likelihood Estimation}
\label{subsec: RCMLe}
The \gls{ML} estimator in \cref{eq: load estimator} is computationally complex to implement.
This is mainly due to $P\left(\overline{X} \middle| \overline{H}, \overline{L}, \overline{\Phi}\right)$ in \cref{eq: conditional pattern probability} having a large number of combinations for certain values of $\mathcal{X}$, $N^{h}_{ 1}$, and $N^{l}_{ 1}$.
Thus, we propose an \gls{RCML} estimator that discards $P\left(\overline{X} \middle| \overline{H}, \overline{L}, \overline{\Phi}\right)$ and relies on the remaining factors in \cref{eq: load estimator} to estimate the network load.
In other words, the probability of an access pattern $\bm{\pi}$ is approximated as follows:
\begin{equation}
    \label{eq: suboptimal conditional pattern probability}
    P(\bm{\pi}) \approx \hat{P}(\bm{\pi}) = P\left(\overline{H}, \overline{L}, \overline{\Phi}\right) = 
    P\left(\overline{H}\right) 
    P\left(\overline{L} \middle| \overline{H}\right)
    P\left(\overline{\Phi} \middle| \overline{H}, \overline{L}\right)
\end{equation}

The \gls{RCML} estimator is then obtained by \cref{eq: load estimator} with $\hat{P}(\bm{\pi})$ in place of $P(\bm{\pi})$.

\section{Simulation Results}
\label{sec: simulation}
In this section, we evaluate the performance of the proposed estimators in terms of the \gls{MAE} for varying numbers of \glspl{UE} and the number of patterns observed.
We consider three simulation setups.
In the first setup, all \glspl{UE} are assumed to be in the same priority class; hence, $\bm{p}^h = \bm{p}^l$.
This scenario replicates the \gls{RACH} in LTE and the contention-based \gls{RACH} procedure in 5G, with the \glspl{RB} representing preambles in the random access procedure~\cite{polese2016m2m, morais20245g}.
In the second setup, the \glspl{UE} are split into two priority classes, with the \glspl{H-UE} having more \glspl{RB} available than the \glspl{L-UE}. This setup corresponds to virtually allocating the \glspl{RB} to the \glspl{UE} depending on their priority class~\cite{cheng2011prioritized, lin2014prada}.
Finally, the third setup utilizes non-uniform \gls{RB} access probabilities introducing a generalized approach to the \gls{RACH} protocol, thereby opening up opportunities for integrating different priority classes more effectively.
The results presented are averages of 50 Monte Carlo simulations.
In all simulations, the number of \glspl{RB} is set to $M=6$, $T$ is set to 1, 3, or 10, and the number of \glspl{L-UE} ranges from 0 to 7.

\subsection{Single Priority Class}
\label{subsec: simulation 1}
In the first simulation, $n^h$ is fixed to 2, while the elements in $\bm{p}^h$ and $\bm{p}^l$ are set to $1/M$, indicating that all \glspl{RB} are equally likely to be chosen by any \gls{UE}.
The \gls{MAE} of the \gls{ML} and \gls{RCML} estimators against $n^l$ are shown in \cref{fig: MAE vs. nl}.
As we can see, on average, the \gls{ML} estimator (\cref{fig: exact 1}) outperforms the \gls{RCML} estimator (\cref{fig: imperfect 1}) due to the \gls{RCML} not using all the available pattern information.
As expected, the \gls{MAE} decreases as $T$ increases due to the extra information available to the \gls{BS}.
Furthermore, for a fixed $T$, an increase in the overloading factor leads to a higher number of collisions resulting in a higher \gls{MAE}.
A high collision number in the patterns results in less useful information for the estimators, leading to more errors.
To further drive this point, consider a pattern consisting of only collisions; we cannot distinguish between the \glspl{H-UE} and \glspl{L-UE}.
On the other hand, if a pattern consists of only non-collisions, we can exactly determine the number of \glspl{H-UE} and \glspl{L-UE}.
\setlength{\textfloatsep}{0pt}
\setlength{\floatsep}{0pt}
\setlength{\intextsep}{2.5pt}
\begin{figure}
    \centering
    \subfloat[]{
            \includegraphics{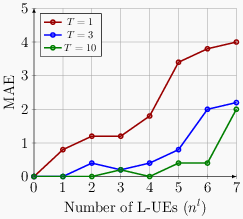}
            \label{fig: exact 1}
    }
    \subfloat[]{
        \includegraphics{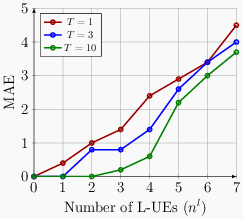}
        \label{fig: imperfect 1}
    }
    \caption{\gls{MAE} vs. $n^l$ for \gls{ML} (\ref{fig: exact 1}) and \gls{RCML} (\ref{fig: imperfect 1}) estimators for different $T$.}
    \label{fig: MAE vs. nl}
\end{figure}

Finally, to illustrate the estimators' operation, \cref{fig: hm} shows a heatmap of the likelihood function for $n^h=2$, $n^l=4$, and $T=1$.
The heatmap shows that the likelihood function is maximized for the above parameters when $\hat{n}^h=2$ and $\hat{n}^l=4$, which are the true number of \glspl{H-UE} and \glspl{L-UE}.
It is noted that the $T=1$ scenario is similar to other load estimation work that utilize one \gls{RACH} slot in their formulations~\cite{sun2018detecting, tello2018efficient}.
\setlength{\textfloatsep}{0pt}
\setlength{\floatsep}{0pt}
\setlength{\intextsep}{2.5pt}
\begin{figure}
    \centering
    \includegraphics{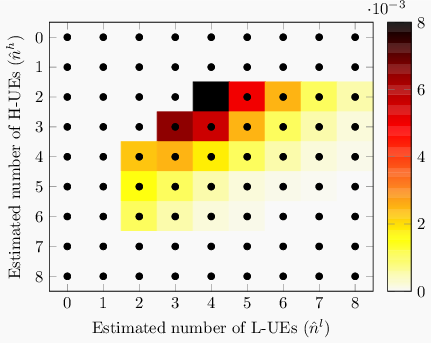}
    \caption{Likelihood heatmap for $n^h=2$, $n^l=4$, and $T=1$.}
    \label{fig: hm}
\end{figure}

\subsection{Two Priority Classes}
\label{subsec: simulation 2}
In the second simulation, $n^h$ is set to 1 or 2.
Furthermore, the elements in $\bm{p}^h$ are fixed to $1/M$, while $\bm{p}^l = [1/3, 1/3, 1/3, 0, 0, 0]$, indicating that the \glspl{H-UE} have twice as many \glspl{RB} available as the \glspl{L-UE}.
\cref{fig: MAE vs. nl 2} shows the \gls{MAE} against $n^l$ for different values of $n^h$ and $T$ for both estimators.
For a fixed $n^h$, the \gls{MAE} decreases as $T$ increases.
On the other hand, for a fixed $T$, the \gls{MAE} increases as $n^h$ increases since the overloading factor increases.
\setlength{\textfloatsep}{0pt}
\setlength{\floatsep}{0pt}
\setlength{\intextsep}{2.5pt}
\begin{figure}
    \centering
    \subfloat[]{
        \includegraphics{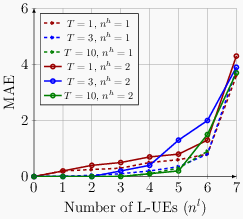}
        \label{fig: exact 2}
    }
    \subfloat[]{
        \includegraphics{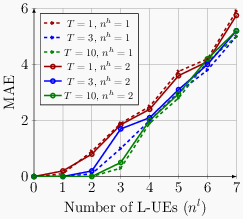}
        \label{fig: Imperfect 2}
    }
    \caption{\gls{MAE} vs. $n^l$ for \gls{ML} (\ref{fig: exact 2}) and \gls{RCML} (\ref{fig: Imperfect 2}) estimators for different $T$ and $n^h$.}
    \label{fig: MAE vs. nl 2}
\end{figure}

\subsection{Two Priority Classes with Non-Uniform RB Allocation}
\label{subsec: simulation 3}
In the third simulation, we set $\bm{p}^h = [1/12, 1/12, 2/12, 2/12, 3/12, 3/12]$ and $\bm{p}^l = [4/12, 3/12, 2/12, 1/12, 1/12, 1/12]$.
\cref{fig: MAE vs. nl 3} showcases the \gls{MAE} for the estimators, which follow a similar trend to the results of \cref{subsec: simulation 2}.
\setlength{\textfloatsep}{2.5pt}
\setlength{\floatsep}{2.5pt}
\setlength{\intextsep}{2.5pt}
\begin{figure}
    \centering
    \subfloat[]{
        \includegraphics{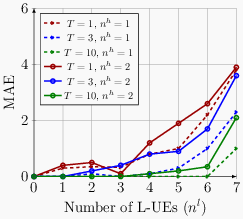}
        \label{fig: exact 3}
    }
    \subfloat[]{
        \includegraphics{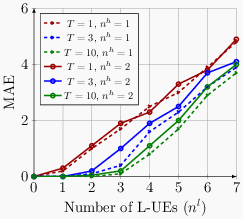}
        \label{fig: Imperfect 3}
    }
    \caption{\gls{MAE} vs. $n^l$ for \gls{ML} (\ref{fig: exact 3}) and \gls{RCML} (\ref{fig: Imperfect 3}) estimators for different $T$ and $n^h$ using non-uniform $\bm{p}^h$ and $\bm{p}^l$.}
    \label{fig: MAE vs. nl 3}
\end{figure}
We note that for a relatively low $n^l$, the \gls{RCML} estimator performs similarly compared to the \gls{ML} estimator but starts to diverge as $n^l$ increases.
Hence the last factor in \cref{eq: conditional pattern probability} becomes more significant as the network load increases.
Additionally, the \gls{RCML} estimator is significantly faster than the \gls{ML} estimator; on average, across all three simulations, the \gls{RCML} estimator was 46 times faster than the \gls{ML} estimator.
Furthermore, there is a tradeoff between the estimators' accuracy and the time taken to estimate the load; the larger $T$ is, the higher the accuracy.

Finally, comparing the results in \cref{fig: MAE vs. nl 2} and \cref{fig: MAE vs. nl 3} we see that utilizing non-uniform \gls{RB} selection probabilities leads to increased load estimation accuracy.
It is important to note that the computational complexity in both cases is the same.

\section{Conclusion}
\label{sec: conclusion}
In this paper, we proposed two network load estimators for a two-priority \gls{mMTC} \gls{RACH} that operate by observing multiple \gls{RACH} slots.
First, we modelled the \gls{RACH} slots by \gls{RB} access patterns observed by the \gls{BS}.
Then, we formulated the likelihood of the number of high and low priority \glspl{UE} given the observed access patterns.
Based on the formulated likelihood function, we developed a \gls{ML} estimator and a \gls{RCML} estimator, where the latter approximated the likelihood function by discarding a portion of the pattern information.
Finally, we conducted numerical simulations to evaluate the performance of the proposed estimators.
The results showed that utilizing non-uniform \gls{RB} selection probabilities may lead to increased load estimation accuracy.
Importantly, these accuracy gains come at no additional computational cost.
Additionally, for a low overloading factor, we found that the \gls{RCML} estimator offered comparable performance to its \gls{ML} counterpart at significant reduction of computational cost.

Our work can be further extended along several dimensions.
For example, by considering the arrival rate of \glspl{UE} in each priority class, or by considering more than two priority classes.
Moreover, we aim to expand the system model to more realistically represent bursty mMTC traffic by adopting the traffic model developed by~\cite{3GPP}.

\appendices
\setlength{\abovedisplayskip}{3pt}
\setlength{\belowdisplayskip}{3pt}
\section{Proof of \texorpdfstring{\cref{prop: P(H)}}{Lg}}
\label{app: P(H) proof}

    The upcoming derivations of the probabilities in \cref{eq: conditional pattern probability} require the use of the following lemma.
    \begin{lemma}
        \label{lemma: normalization}
        Consider the mutually exclusive and exhaustive events $A_1, \dots, A_n$ with the corresponding probabilities $\pi_1, \dots, \pi_n$.
        Given that the event $A_k$ does not occur, 
        the probability of each remaining event $A_i$, $i\not= k$, is given by:
        \begin{equation}
            \label{eq: normalization}
            \hat{\pi}_i = \frac{\pi_i}{\sum_{\substack{j=1\\ j \neq k}}^{n} \pi_j}, \quad \forall i \in \{1, \dots, n\} \setminus \{k\}
        \end{equation}
    \end{lemma}
    \begin{proof}
        Let $\mathcal{A} = \{A_1, A_2, \dots, A_n\}$. 
        Given that event $A_k$ does not occur, the conditional probability of each remaining event $A_i \in \mathcal{A} \setminus \{A_k\}$ is given by:
        \begin{equation}
        P(A_i | A^c_k) = \frac{P(A_i \cap A^c_k)}{P(A^c_k)}
        \end{equation}
        where $A^c_k$ is the complement of $A_k$ with $P(A^c_k) = 1 - \pi_k = \sum_{i=1, i \neq k}^{M} \pi_i$. 
        The intersection of each remaining event $A_i$ with $A^c_k$ is the event $A_i$ itself since all remaining events are contained within $A^c_k$:
        $$P(A_i \cap A^c_k) = P(A_i).$$
        Hence, the probability of each remaining event $A_i$ given event $A_k$ does not occur is:
        \begin{multline}
            \hat{\pi}_i = P(A_i | A^c_k) = P(A_i \cap A^c_k) / P(A^c_k)\\
            = P(A_i) / P(A^c_k) = \pi_i / \sum_{\substack{i=1\\ i \neq k}}^{n} \pi_i, \quad \forall A_i \in \mathcal{A} \setminus \{A_k\}
        \end{multline}
        such that $\sum_{i=1}^{n} \hat{\pi}_i = 1$.
    \end{proof}
    We can now begin the proof of \cref{prop: P(H)}, as shown below.
        The first factor in \cref{eq: P(H)} is the probability that only one \gls{H-UE} occupies RB $h_1$, and no other \glspl{UE} are in RB $h_1$:
        \begin{equation}
            \label{eq: P(h_1) proof}
            P(H_{h_1}) =
            \binom{n^h}{1} p^h_{h_1} (1-p^h_{h_1})^{n^h-1} \binom{n^l}{0} (p^l_{h_1})^{0} (1-p^l_{h_1})^{n^l},
        \end{equation}
        which when simplified results in \mbox{\cref{eq: P(h_1)}}. 
        The second factor in \cref{eq: P(H)} is the probability that only one \gls{H-UE} occupies RB $h_2$, and no other \glspl{UE} are in RB $h_2$, given that one \gls{H-UE} is in RB $h_1$.
        Thus RB $h_1$ can no longer be chosen by any other \gls{UE}, i.e., $p^h_{h_1}=p^l_{h_1}=0$; we have:
    \begin{multline}
        \label{eq: P(h_2|h_1)}
        P(H_{h_2}|H_{h_1}) =
        \binom{N^{h}_{2}}{1} \hat{p}^h_{h_2} (1-\hat{p}^h_{h_2})^{N^{h}_{2}-1} \times \\
        \binom{n^l}{0} (\hat{p}^l_{h_2})^{0} (1-\hat{p}^l_{h_2})^{n^l}
        = N^{h}_{2} \hat{p}^h_{h_2} (1-\hat{p}^h_{h_2})^{N^{h}_{2}-1} (1-\hat{p}^l_{h_2})^{n^l}
    \end{multline}
    where from \mbox{\cref{eq: normalization P(H)0}}, $N^{h}_{2} = n^h - 1$ is the remaining \glspl{H-UE}.
    $\hat{p}^h_{h_2}$ and $\hat{p}^l_{h_2}$ are the probabilities that an \gls{H-UE} and an \gls{L-UE} can occupy RB $h_2$ given that RB $h_1$  is unavailable, which are given by:
    \begin{align}
        \hat{p}^h_{h_2} = \frac{p^h_{h_2}}{\sum_{\substack{i=1 \\ i \not= h_1}}^M p^h_{i}}, && \hat{p}^l_{h_2} = \frac{p^l_{h_2}}{\sum_{\substack{i=1 \\ i \not= h_1}}^M p^l_{i}}
    \end{align}
    In general, the probability that an \gls{H-UE} occupies RB $h_k$ given that RBs $h_1, \dots, h_{k-1}$ are unavailable is:
    \begin{multline}
        \label{eq: P(h_k|h_1 ... h_{k-1}) proof}
        P(H_{h_k}|H_{h_1}, \dots, H_{h_{k-1}}) =
        \binom{N^{h}_{k}}{1} \hat{p}^h_{h_k} (1-\hat{p}^h_{h_k})^{N^{h}_{k}-1} \times \\
        \binom{n^l}{0} (\hat{p}^l_{h_k})^{0} (1-\hat{p}^l_{h_k})^{n^l}
    \end{multline}
    where $N^{h}_{k}$ is obtained from \mbox{\cref{eq: normalization P(H)0}}.
    While $\hat{p}^h_{h_k}$ and $\hat{p}^l_{h_k}$ are the probabilities that an \gls{H-UE} and an \gls{L-UE} can occupy RB $h_k$ given that RBs $h_1$, \dots, $h_{k-1}$ are unavailable, which are calculated using \mbox{\cref{eq: normalization p(H)1}} and \mbox{\cref{eq: normalization p(H)2}}, respectively.
    Simplifying \mbox{\cref{eq: P(h_k|h_1 ... h_{k-1}) proof}}, results in \mbox{\cref{eq: P(h_k|h_1 ... h_{k-1})}}, which completes the proof.

\section{Recursive Algorithms}
\label{app: rec algs}
\cref{alg: P(H),alg: P(L|H),alg: P(X|H L E) recursive} present a recursive procedure to calculate $P\left(\overline{H}\right)$, $P\left(\overline{L} \middle| \overline{H}\right)$ and $P\left(\overline{X} \middle| \overline{H}, \overline{L}, \overline{\Phi}\right)$, respectively.

Notice that in \cref{alg: P(X|H L E) recursive} we store values for all the possible $N^{h}_{j}$, $N^{l}_{j}$, $\hat{p}^{h}_{X_j}$, and $\hat{p}^{l}_{X_j}$; instead, only the values for the current iteration are calculated then overwritten with the values of the next iteration.
Additionally, we use two helper functions, \textsc{FindCombinations} and \textsc{FindCombinationsH}, which find all feasible combinations of the total number of \glspl{UE} and \glspl{H-UE} that can occupy the \glspl{RB} in $\mathcal{X}$, respectively.

\setlength{\textfloatsep}{0pt}
\setlength{\floatsep}{0pt}
\setlength{\intextsep}{2.5pt}
\begin{algorithm}
    \small
    \caption{$P\left(\overline{H}\right)$}
    \label{alg: P(H)}
    \begin{algorithmic}[1]
        \State Initialization: $z = 1$, $N^h=n^h$, $\hat{\bm{p}}^h=\bm{p}^h$, and $\mathbf{\hat{p}}^l=\bm{p}^l$ \label{alg: P(H) initialization}
        \For{each $i$ in $\mathcal{H}$} \label{alg: P(H) for loop}
            \State $z \gets z \times N^h \hat{p}^h_{i} (1-p^h_{i})^{N^h-1} (1-\hat{p}^l_{i})^{n^l}$ \label{alg: P(H) z update}
            \State $N^h \gets N^h - 1$ \label{alg: P(H) N^h update}
            \State $\hat{p}^h_{i} \gets 0, \qquad \hat{p}^l_{i} \gets 0$ \label{alg: P(H) p^h update}
            \State $\hat{p}^h_{j} \gets \frac{\hat{p}^h_{j}}{\sum_{\substack{k=1\\ k \neq i}}^M \hat{p}^h_{k}}, \quad \hat{p}^l_{j} \gets \frac{\hat{p}^l_{j}}{\sum_{\substack{k=1 \\ k \neq i}}^M \hat{p}^l_{k}} \text{ for } j \neq i$ \label{alg: P(H) p^h update2}
        \EndFor
        \State At the end of the $i$th iteration $z$ holds $P\left(\overline{H}\right)$.
    \end{algorithmic}
\end{algorithm}

\setlength{\textfloatsep}{0pt}
\setlength{\floatsep}{0pt}
\setlength{\intextsep}{1pt}

\begin{algorithm}
    \small
    \caption{$P\left(\overline{L} \middle| \overline{H}\right)$}
    \label{alg: P(L|H)}
    \begin{algorithmic}[1]
        \State Initialization: $z=1$, $N^h=n^h-H$, $N^l=n^l$ \label{alg: P(L|H) initialization}
        \State $\hat{p}^h_{j} = \begin{cases}
            0 & \text{if } j \in \mathcal{H} \\
            \frac{p^h_{j}}{\sum_{\substack{i\notin \mathcal{H}}} p^h_{i}} & \text{otherwise}
        \end{cases}$ \label{alg: P(L|H) p^h initialization}
        \State $\hat{p}^l_{j} = \begin{cases}
            0 & \text{if } j \in \mathcal{H} \\
            \frac{p^l_{j}}{\sum_{\substack{i\notin \mathcal{H}}} p^l_{i}} & \text{otherwise}
        \end{cases}$ \label{alg: P(L|H) p^l initialization}
        \For{each $i$ in $\mathcal{L}$} \label{alg: P(L|H) for loop}
            \State $z \gets z \times N^l \hat{p}^l_{i} (1-\hat{p}^l_{i})^{N^l-1} (1-\hat{p}^h_{i})^{N^h}$ \label{alg: P(L|H) z update}
            \State $N^l \gets N^l - 1$ \label{alg: P(L|H) N^l update}
            \State $\hat{p}^h_{i} \gets 0, \qquad \hat{p}^l_{i} \gets 0$ \label{alg: P(L|H) p^h update}
            \State $\hat{p}^h_{j} \gets \frac{\hat{p}^h_{j}}{\sum_{\substack{k=1\\ k \neq i}}^M \hat{p}^h_{k}}, \quad \hat{p}^l_{j} \gets \frac{\hat{p}^l_{j}}{\sum_{\substack{k=1 \\ k \neq i}}^M \hat{p}^l_{k}} \text{ for } j \neq i$ \label{alg: P(L|H) p^h update2}
        \EndFor
        \State At the end of the $i$th iteration $z$ holds $P\left(\overline{L} \middle| \overline{H}\right)$.
    \end{algorithmic}
\end{algorithm}

\setlength{\textfloatsep}{0pt}
\setlength{\floatsep}{0pt}
\setlength{\intextsep}{1pt}
\begin{algorithm}
    \small
    \caption{$P\left(\overline{X} \middle| \overline{H}, \overline{L}, \overline{\Phi}\right)$}
    \label{alg: P(X|H L E) recursive}
    \begin{algorithmic}[1]
        \State Initialization: $z=0$, $N^{h}=n^h-H$, $N^{l}=n^l-L$, $N = N^{h} + N^{l}$ \label{alg: P(X|H L E) initialization}
        \State $\hat{p}^h_{j} = \begin{cases}
            0 & \text{if } j \in \mathcal{H}, \mathcal{L}, \mathbf{\mathit{\Phi}} \\
            \frac{p^h_{j}}{\sum_{\substack{i\in \mathcal{X}} p^h_{i}}} & \text{otherwise}
        \end{cases}$ \label{alg: P(X|H L E) p^h initialization}
        \State $\hat{p}^l_{j} = \begin{cases}
            0 & \text{if } j \in \mathcal{H}, \mathcal{L}, \mathbf{\mathit{\Phi}} \\
            \frac{p^l_{j}}{\sum_{\substack{i\in \mathcal{X}}} p^l_{i}} & \text{otherwise}
        \end{cases}$  \label{alg: P(X|H L E) p^l initialization}
        \State $\mathcal{K} \gets \textsc{FindCombinations}(X, N, \text{empty list})$ \label{alg: P(X|H L E) find combinations}
        \For{each $\mathbf{k}$ in $\mathcal{K}$} \label{alg: P(X|H L E) for loop k}
            \State $\mathcal{I} \gets \textsc{FindCombinationsH}(X, N^{h}, \mathbf{k}, \text{empty list})$ \label{alg: P(X|H L E) find combinations H}
            \For{each $\mathbf{i}$ in $\mathcal{N}$} \label{alg: P(X|H L E) for loop i}
                \State Initialization: $z^{\prime}=1$, $p=1$, $N^{h}$, $N^{l}$, $\hat{p}^h_{j}$, $\hat{p}^l_{j}$\label{alg: P(X|H L E) initialization 2}
                \For{ each $j$ in $\mathcal{X}$} \label{alg: P(X|H L E) for loop j}
                    \State $z^{\prime} \gets z^{\prime} \times 
                    \binom{N^h}{i_p} (\hat{p}^h_{j})^{i_p} (1-\hat{p}^h_{j})^{N^h-i_p} \times \binom{N^l}{k_p-i_p} (\hat{p}^l_{j})^{k_p-i_p} (1-\hat{p}^l_{j})^{N^l-k_p+i_p}$ \label{alg: P(X|H L E) z' update}
                    \State $N^h \gets N^h-i_p, \qquad N^l \gets N^l-(k_p-i_p)$ \label{alg: P(X|H L E) N^h update}
                    \State $\hat{p}^h_{j} \gets 0, \qquad \hat{p}^l_{j} \gets 0$ \label{alg: P(X|H L E) p^h update}
                    \State $\hat{p}^h_{m} \gets \frac{\hat{p}^h_{m}}{\sum_{\substack{k=1\\ k \neq j}}^M \hat{p}^h_{k}}, \quad \hat{p}^l_{m} \gets \frac{\hat{p}^l_{m}}{\sum_{\substack{k=1 \\ k \neq j}}^M \hat{p}^l_{k}} \text{ for } m \neq j$ \label{alg: P(X|H L E) p^h update2}
                    \State $p \gets p+1$ \label{alg: P(X|H L E) p update}
                \EndFor
                \State $z \gets z + z^{\prime}$ \label{alg: P(X|H L E) z update2}
            \EndFor
        \EndFor
        \State At the end of the $k$th iteration $z$ holds $P\left(\overline{X} \middle| \overline{H}, \overline{L}, \overline{\Phi}\right)$.
    \end{algorithmic}
\end{algorithm}

\bibliographystyle{IEEEtran}
\bibliography{references}

\begin{thebibliography}{10}
\providecommand{\url}[1]{#1}
\csname url@samestyle\endcsname
\providecommand{\newblock}{\relax}
\providecommand{\bibinfo}[2]{#2}
\providecommand{\BIBentrySTDinterwordspacing}{\spaceskip=0pt\relax}
\providecommand{\BIBentryALTinterwordstretchfactor}{4}
\providecommand{\BIBentryALTinterwordspacing}{\spaceskip=\fontdimen2\font plus
\BIBentryALTinterwordstretchfactor\fontdimen3\font minus
  \fontdimen4\font\relax}
\providecommand{\BIBforeignlanguage}[2]{{%
\expandafter\ifx\csname l@#1\endcsname\relax
\typeout{** WARNING: IEEEtran.bst: No hyphenation pattern has been}%
\typeout{** loaded for the language `#1'. Using the pattern for}%
\typeout{** the default language instead.}%
\else
\language=\csname l@#1\endcsname
\fi
#2}}
\providecommand{\BIBdecl}{\relax}
\BIBdecl

\bibitem{leyva2016performance}
I.~Leyva-Mayorga, L.~Tello-Oquendo, V.~Pla, J.~Martinez-Bauset, and
  V.~Casares-Giner, ``Performance analysis of access class barring for handling
  massive {M2M} traffic in {LTE-A} networks,'' in \emph{2016 {IEEE}
  international conference on communications (icc)}.\hskip 1em plus 0.5em minus
  0.4em\relax IEEE, 2016, pp. 1--6.

\bibitem{condoluci2015toward}
M.~Condoluci, M.~Dohler, G.~Araniti, A.~Molinaro, and K.~Zheng, ``Toward {5G}
  densenets: architectural advances for effective machine-type communications
  over femtocells,'' \emph{{IEEE} Communications Magazine}, vol.~53, no.~1, pp.
  134--141, 2015.

\bibitem{wiriaatmadja2014hybrid}
D.~T. Wiriaatmadja and K.~W. Choi, ``Hybrid random access and data transmission
  protocol for machine-to-machine communications in cellular networks,''
  \emph{{IEEE} Transactions on Wireless Communications}, vol.~14, no.~1, pp.
  33--46, 2014.

\bibitem{wei2014modeling}
C.-H. Wei, G.~Bianchi, and R.-G. Cheng, ``Modeling and analysis of random
  access channels with bursty arrivals in {OFDMA} wireless networks,''
  \emph{{IEEE} Transactions on Wireless communications}, vol.~14, no.~4, pp.
  1940--1953, 2014.

\bibitem{arouk2016accurate}
O.~Arouk, A.~Ksentini, and T.~Taleb, ``How accurate is the {RACH} procedure
  model in {LTE} and {LTE-A}?'' in \emph{2016 International Wireless
  Communications and Mobile Computing Conference (IWCMC)}.\hskip 1em plus 0.5em
  minus 0.4em\relax IEEE, 2016, pp. 61--66.

\bibitem{tello2018efficient}
L.~Tello-Oquendo, V.~Pla, I.~Leyva-Mayorga, J.~Martinez-Bauset,
  V.~Casares-Giner, and L.~Guijarro, ``Efficient random access channel
  evaluation and load estimation in {LTE-A} with massive {MTC},'' \emph{{IEEE}
  Transactions on Vehicular Technology}, vol.~68, no.~2, pp. 1998--2002, 2018.

\bibitem{sun2018detecting}
J.~Sun, R.~Liu, and E.~Paolini, ``Detecting the number of active users in
  {IRSA} access protocols,'' in \emph{2018 {IEEE} 29th Annual International
  Symposium on Personal, Indoor and Mobile Radio Communications (PIMRC)}.\hskip
  1em plus 0.5em minus 0.4em\relax IEEE, 2018, pp. 1972--1976.

\bibitem{duan2016d}
S.~Duan, V.~Shah-Mansouri, Z.~Wang, and V.~W. Wong, ``{D-ACB}: Adaptive
  congestion control algorithm for bursty {M2M} traffic in {LTE} networks,''
  \emph{{IEEE} Transactions on Vehicular Technology}, vol.~65, no.~12, pp.
  9847--9861, 2016.

\bibitem{access2021medium}
{Evolved Universal Terrestrial Radio Access}, ``Medium access control ({MAC})
  protocol specification ({3GPP} {TS} 36.321 version 8.5. 0 release 8),''
  \emph{ETSI TS}, vol. V16, no. 321, 2021.

\bibitem{althumali2020dynamic}
H.~D. Althumali, M.~Othman, N.~K. Noordin, and Z.~M. Hanapi, ``Dynamic backoff
  collision resolution for massive {M2M} random access in cellular {IoT}
  networks,'' \emph{{IEEE} Access}, vol.~8, pp. 201\,345--201\,359, 2020.

\bibitem{chowdhury2022queue}
M.~R. Chowdhury and S.~De, ``Queue-aware access prioritization for massive
  machine-type communication,'' \emph{{IEEE} Internet of Things Journal},
  vol.~9, no.~17, pp. 15\,858--15\,873, 2022.

\bibitem{althumali2022priority}
H.~Althumali, M.~Othman, N.~K. Noordin, and Z.~M. Hanapi, ``Priority-based
  load-adaptive preamble separation random access for {QoS}-differentiated
  services in {5G} networks,'' \emph{Journal of Network and Computer
  Applications}, vol. 203, p. 103396, 2022.

\bibitem{tello2017performance}
L.~Tello-Oquendo, I.~Leyva-Mayorga, V.~Pla, J.~Martinez-Bauset, J.-R. Vidal,
  V.~Casares-Giner, and L.~Guijarro, ``Performance analysis and optimal access
  class barring parameter configuration in {LTE-A} networks with massive {M2M}
  traffic,'' \emph{{IEEE} Transactions on Vehicular Technology}, vol.~67,
  no.~4, pp. 3505--3520, 2017.

\bibitem{alvi2021performance}
M.~Alvi, K.~M. Abualnaja, W.~T. Toor, and M.~Saadi, ``Performance analysis of
  access class barring for next generation {IoT} devices,'' \emph{Alexandria
  Engineering Journal}, vol.~60, no.~1, pp. 615--627, 2021.

\bibitem{duan2013dynamic}
S.~Duan, V.~Shah-Mansouri, and V.~W. Wong, ``Dynamic access class barring for
  {M2M} communications in {LTE} networks,'' in \emph{2013 {IEEE} Global
  Communications Conference (GLOBECOM)}.\hskip 1em plus 0.5em minus 0.4em\relax
  IEEE, 2013, pp. 4747--4752.

\bibitem{zhang2022ppo}
H.~Zhang, M.~Jiang, X.~Liu, X.~Wen, N.~Wang, and K.~Long, ``{PPO}-based {PDACB}
  traffic control scheme for massive {IoV} communications,'' \emph{{IEEE}
  Transactions on Intelligent Transportation Systems}, 2022.

\bibitem{wu2012fast}
H.~Wu, C.~Zhu, R.~J. La, X.~Liu, and Y.~Zhang, ``Fast adaptive {S-ALOHA} scheme
  for event-driven machine-to-machine communications,'' in \emph{2012 {IEEE}
  vehicular technology conference (VTC Fall)}.\hskip 1em plus 0.5em minus
  0.4em\relax IEEE, 2012, pp. 1--5.

\bibitem{galinina2013stabilizing}
O.~Galinina, A.~Turlikov, S.~Andreev, and Y.~Koucheryavy, ``Stabilizing
  multi-channel slotted aloha for machine-type communications,'' in \emph{2013
  {IEEE} International Symposium on Information Theory}.\hskip 1em plus 0.5em
  minus 0.4em\relax IEEE, 2013, pp. 2119--2123.

\bibitem{wang2015optimal}
Z.~Wang and V.~W. Wong, ``Optimal access class barring for stationary machine
  type communication devices with timing advance information,'' \emph{{IEEE}
  Transactions on Wireless communications}, vol.~14, no.~10, pp. 5374--5387,
  2015.

\bibitem{polese2016m2m}
M.~Polese, M.~Centenaro, A.~Zanella, and M.~Zorzi, ``{M2M} massive access in
  {LTE}: {RACH} performance evaluation in a smart city scenario,'' in
  \emph{2016 {IEEE} International Conference on Communications (ICC)}.\hskip
  1em plus 0.5em minus 0.4em\relax ieee, 2016, pp. 1--6.

\bibitem{morais20245g}
D.~H. Morais, ``{5G NR} overview and physical layer,'' in \emph{Key
  5G/5G-Advanced Physical Layer Technologies: Enabling Mobile and Fixed
  Wireless Access}.\hskip 1em plus 0.5em minus 0.4em\relax Springer, 2024, pp.
  233--321.

\bibitem{cheng2011prioritized}
J.-P. Cheng, C.-h. Lee, and T.-M. Lin, ``Prioritized random access with dynamic
  access barring for {RAN} overload in {3GPP} {LTE-A} networks,'' in \emph{2011
  {IEEE} GLOBECOM Workshops (GC Wkshps)}.\hskip 1em plus 0.5em minus
  0.4em\relax IEEE, 2011, pp. 368--372.

\bibitem{lin2014prada}
T.-M. Lin, C.-H. Lee, J.-P. Cheng, and W.-T. Chen, ``{PRADA}: Prioritized
  random access with dynamic access barring for {MTC} in {3GPP} {LTE-A}
  networks,'' \emph{{IEEE} Transactions on Vehicular Technology}, vol.~63,
  no.~5, pp. 2467--2472, 2014.

\bibitem{3GPP}
3GPP, ``Study on {RAN} improvements for machine type communications,'' TR
  37.868, Tech. Rep., 2011.

\end{thebibliography}
\end{document}